\documentclass[12pt]{article}
 
\usepackage{amsfonts}
\usepackage{amssymb}
\usepackage{amsmath}  
\usepackage{graphicx, comment, color, hyperref}
\usepackage{authblk}

\usepackage[a4paper]{geometry}
\newtheorem{thm}{Theorem}

\newtheorem{lemma}[thm]{Lemma}

\newtheorem{defn}[thm]{Definition}

\newenvironment{proof}{{\noindent{\em Proof:}}}{$\hfill\Box$}

\def\tr{\hbox{Tr}}

\def\be{\begin{eqnarray}}
\def\ee{\end{eqnarray}}
\def\bee{\begin{eqnarray*}}
\def\eee{\end{eqnarray*}}
\def\bmx{\begin{pmatrix}}
\def\emx{\end{pmatrix}}

\def\ts{\textstyle}

\def\rt2{\ts \frac{1}{\sqrt{2}} }

\def\ot{\otimes}

\newcommand{\cD}{\mathcal{D}}
\newcommand{\cB}{\mathcal{B}}
\newcommand{\cH}{\mathcal{H}}

\newcommand{\cI}{\mathcal{I}}
\newcommand{\dd}{\text{\rm{d}}}
\newcommand{\bL}{\mathbf{L}}

\title{Hypercontractivity and the logarithmic Sobolev inequality for the completely bounded norm}

\author[1,2]{Salman Beigi}
\author[3]{Christopher King}
\affil[1]{\it \small School of Mathematics, Institute for Research in Fundamental Sciences (IPM), Tehran, Iran}
\affil[2]{\it \small Department of Information Engineering, The Chinese University of Hong Kong, Hong Kong}
\affil[3]{\it \small Department of Mathematics, Northeastern University, Boston MA 02115}

\begin{document}

\maketitle

\begin{abstract}
We develop the notions of hypercontractivity (HC) and the log-Sobolev (LS) inequality for completely bounded norms of one-parameter semigroups of super-operators acting on matrix algebras. We prove the equivalence of the completely bounded versions of HC and LS under suitable hypotheses. We also prove a version of the Gross Lemma which allows LS at general $q$ to be deduced from LS at $q=2$. 
\end{abstract}

\section{Introduction}
The notions of hypercontractivity (HC) and the logarithmic Sobolev (LS) inequalities were originally introduced in the context of
quantum field theory~\cite{Nel1,Simon,Gross2}. The  HC inequality can be formulated as follows: 
for $1 \le q \le p$, and for a suitable operator $A$,
\bee
\| e^{ - t A} \|_{q \rightarrow p} = \sup_{\| f \|_q \le 1} \| e^{- t A} f \|_p \le 1,
\quad \text{ if and only if } \quad t \ge \frac{1}{2} \log \left( \frac{p-1}{q-1} \right).
\eee
The related concept of the logarithmic Sobolev inequality is an infinitesimal version of 
hypercontractivity, obtained by setting $q=2$, $p(t) = 1 + e^{2 t}$ and taking the derivative
at $t=0$. In the original quantum field theory setting, $A$ was the Hamiltonian for the free bosonic field in two
spacetime dimensions, and $f$ was a state in the bosonic Fock space. These results were later
extended to the case of the free fermion field~\cite{Gross1, LiebCarlen93}.

Recently, HC and LS~inequalities have found applications in quantum information theory. 
For such applications, $A = {\cal L}$ is often the generator of a one-parameter
semigroup of completely positive maps on an open quantum system, representing its dissipative evolution
in the memoryless (Markovian) approximation, and $f$ is an observable on the system.
In this setting the norm $\| \cdot \|_q$ is usually a Schatten norm on a matrix algebra.
The theory of HC and LS~inequalities in this setting has been developed by Olkiewicz and Zegarlinski~\cite{OZ},
and more recently by Kastoryano and Temme~\cite{KT1}, who also used these methods to derive
mixing time bounds for a variety of quantum channel semigroups~\cite{KT2}.

Many of the results derived for classical Markov chains using HC and LS~inequalities can be extended to
quantum channel semigroups. One notable exception is the `tensoring up' property.
This is the issue of finding bounds for products of independent copies of channels (super-operators), and is concerned
with norms of the type $\| e^{- t_1 {\cal L}_1} \ot e^{- t_2 {\cal L}_2} \|_{q \rightarrow p}$. For classical channels this operator norm
is multiplicative, and thus the `time to contraction' for a product of channels is the maximum of the times to contraction for each individual channel.
For quantum channels this need not be true (although at this time there is no explicit example known of a channel
semigroup which violates this classical `additivity' result, it is widely believed that violations are generic).

In order to handle this non-additivity, one approach is to use a different norm for which additivity is guaranteed, namely the 
completely bounded (CB) norm. The CB~norm \newline $\| \cdot \|_{CB, q \rightarrow p}$ was introduced by Pisier~\cite{Pisier}
and reviews can be found in~\cite{DB14, DJKR}. This norm satisfies the property
\be\label{eq:CB-norm-multiplicative}
{\| \Phi_1 \ot \Phi_2 \|}_{CB, q \rightarrow p} = {\| \Phi_1 \|}_{CB, q \rightarrow p} \, {\| \Phi_2 \|}_{CB, q \rightarrow p},
\ee
for all $p,q \ge 1$ and all completely positive maps $\Phi_1$ and  $\Phi_2$.
One special case of the CB~norm is the well-known diamond norm, which is the case $q=p=1$.

Applying the CB~norm to a semigroup
$e^{- t {\cal L}}$ of completely positive maps, we can investigate the time to contraction and derive the corresponding LS~inequality. In this case, by the above multiplicativity, the time to contraction (under CB~norm) of a product of channels can be computed in terms of the time to contraction of individual ones, as in the classical setting. 

In this paper we introduce LS~inequalities associated with CB~norms. 
We show equivalence of the HC condition and the corresponding
LS~inequality for the CB~norm of quantum channel semigroups. 
The proof requires some novel ingredients which are not required 
in the usual setting of the Schatten matrix norms. We also establish the CB version of the Gross Lemma,
which allows LS for general $q$ to be deduced from LS at $q=2$. Furthermore, we show how this leads to an
`additivity' result for the LS constants of a product channel.

The rest of this paper is organized as follows. We first establish notation, and review the definitions of the
CB~norm. We then state our main result which is a formulation of the LS~inequality for the CB~norm.
This LS~inequality displayed in~\eqref{CB-logSob2} is strongly similar to the usual LS~inequality,
but with the partial trace appearing in some places. The following sections contain our analysis of the CB~norm,
which requires some careful characterization of the minimizers appearing in the definitions.
The Appendix contains some technical results.

%******************************************************************************
\section{Preliminaries}
We label systems and Hilbert spaces by uppercase letters such as $\cH_R$, $\cH_S$,
and denote the tensor product $\cH_R\otimes \cH_S$ by $\cH_{RS}$.
All the spaces considered throughout this paper will be finite dimensional, and
we will use the notation $d_R=\dim\mathcal H_R < \infty$.
The space of linear operators acting on $\cH_R$ will be denoted by $\bL(\cH_R)$, 
and we will often attach a label to an operator $X_R \in \bL(\cH_R)$ to indicate the underlying space.
The adjoint of $X_R$ is denoted by $X_R^*$.
For $X_R\in \bL(\cH_R)$ the \emph{normalized} trace is defined by
\bee
\tau(X_R) = d_R^{-1} \, \tr (X_R).
\eee
We will work mostly with $\tau(\cdot)$ rather than the unnormalized matrix trace,
and this will enter the various definitions of norms that we will use.
Thus the (normalized) $p$-Schatten norm of $X_R\in \bL(\cH_R)$ is defined by
\bee
\| X_R \|_p = \big( \tau (|X_R|^p) \big)^{1/p} = d_R^{-1/p} \, \big( \tr (|X_R|^p) \big)^{1/p}, \quad p \ge 1,
\eee
where $|X_R|=(X_R^* X_R)^{1/2}$.

For an operator $X_R\in \bL(\cH_R)$, we write $X_R\geq 0$ to indicate that $X_R$ is positive semidefinite,
and  $X_R>0$ to indicate that $X_R$ is positive definite.  We denote by ${\cD}^{+}_R\subset \bL(\cH_R)$ 
the set of positive definite matrices normalized with respect to~$\tau$, that is,
\be\label{eq:def-D-R-+}
{\cD}^{+}_{R} = \{ \sigma_{R} \in \bL(\cH_R) \, :\, \sigma_R > 0, \,\, \tau (\sigma_R) = 1 \}.
\ee
For $\sigma_R\in \cD_R^+$ and $\epsilon>0$ we define
\be\label{def:ball}
\mathcal{B}_\epsilon(\sigma_R) = \big\{ \xi \in \cD_R^+:\,   \|\sigma-\xi\|_1 \leq \epsilon   \big\}
= \cD_R^+ \cap \big\{ X_R\in \bL(\cH_R) :\,   \|\sigma- X_R \|_1 \leq \epsilon   \big\}.
\ee

\subsection{Non-commutative $(q, p)$-norm}

For operators $Y_{RS}$ acting on the product space $\cH_{RS}=\cH_R\ot \cH_S$, we will use the non-commutative $(q,p)$-norms introduced
by Pisier \cite{Pisier}, which extend the classical $l^q(l^p)$ norms to bipartite matrices. 
See~\cite{DJKR, DB14} for a review of this notion and its applications in quantum information theory. 
For $1 \le q \le p$ the $(q, p)$-norm is given by
\be\label{def:NC-norm}
\| Y_{RS} \|_{(q,p)} = \inf_{A,B,Z} \big\{\| A_R \|_{2r} \, \| B_R \|_{2r} \, \| Z_{RS} \|_p:\,
Y_{RS} = (A_R \ot I_{S}) Z_{RS} (B_R \ot I_{S}) \big\},
\ee
where $I_{S}\in \bL(\cH_S)$ is the identity operator, and $r$ is given by
\be\label{def:r}
\frac{1}{r} = \frac{1}{q} - \frac{1}{p}.
\ee
Note again that we are using the normalized Schatten norm,
so our definition~\eqref{def:NC-norm} of the $(q, p)$-norm differs from the standard one in~\cite{Pisier} by
an overall multiplicative factor. As shown in~\cite{DJKR},
when $Y$ is positive semidefinite,
without loss of generality we may restrict the infimum in~\eqref{def:NC-norm} to positive definite matrices $A_R=B_R > 0$, 
in which case we have
\bee
Z_{RS} = (A_{R}^{-1} \ot I_{S}) Y_{RS} (A_{R}^{-1} \ot I_{S}).
\eee
Therefore, for positive semidefinite $Y_{RS}\geq 0$ we have
\bee
\| Y_{RS} \|_{(q,p)} = \inf_{A_R > 0}  \,\big\| A_R \big\|_{2r}^2 \cdot \big\| (A_R^{-1} \ot I_{S}) Y_{RS} (A_R^{-1} \ot I_{S}) \big\|_p.
\eee
Moreover, by rescaling we may assume that $\big\| A_R \big\|_{2r} = 1$, which implies that
$A_R^{2 r} \in \cD_R^+$ as defined in~\eqref{eq:def-D-R-+}. So for $1 \le q \le p$  and
$Y_{RS} \ge 0$ we may write
\be\label{NC-norm-inf}
\| Y_{RS} \|_{(q,p)} &=& \inf_{\sigma_R \in {\cD}^{+}_{R}} \,  \big\| (\sigma_R^{-1/2r} \ot I_{S}) Y_{RS} (\sigma_R^{-1/2r} \ot I_{S}) \big\|_p.
\ee

If the ordering of $q,p$ is reversed, so that $1 \le p \le q$, the corresponding expression for the $(q,p)$-norm of a
positive semidefinite operator $Y_{RS}$ is
\be\label{NC-norm-sup}
\| Y_{RS} \|_{(q,p)} &=& \sup_{\sigma_R \in {\cD}^{+}_{R}} \,  \big\| (\sigma_R^{-1/2r} \ot I_{S}) Y_{RS} (\sigma_R^{-1/2r} \ot I_{S}) \big\|_p,
\ee
where again $r$ is given by~\eqref{def:r}, and thus is negative in this case.

When $p=q$ we may use either of the definitions~\eqref{NC-norm-inf} or~\eqref{NC-norm-sup} to compute the $(q, p)$-norm. 
Indeed, for $p=q$ these two definitions coincide and we have $\| Y_{RS} \|_{(q,q)} = \| Y_{RS} \|_q$.

%%%%%%%
\subsection{Completely bounded norm}

For a super-operator $\Phi:\bL(\cH_S) \rightarrow \bL(\cH_{S'})$, the
\emph{completely bounded} (CB) norm is defined by
\be\label{def:CB-norm}
{\| \Phi \|}_{CB, q \rightarrow p} = \sup_{d_R} \, \sup_{Y_{RS}} \,\frac{\| ({\cal I}_R \ot \Phi)(Y_{RS}) \|_{(t,p)}}{\| Y_{RS} \|_{(t,q)}},
\ee
where ${\cal I}_R$ is the identity super-operator acting on $\bL(\cH_R)$, the first supremum is over the dimension $d_R=\dim(\cH_R)$,
and the second supremum is over $Y_{RS} \in \bL(\cH_{RS})$. Moreover, $t\geq 1$ is arbitrary~\cite{Pisier}; the supremum is independent of
the choice of $t$. We will generally use the value $t=q$.

A super-operator $\Phi$ is \emph{positive} if $\Phi(Y_S)\geq 0$ is positive semidefinite for any $Y_S\geq 0$. Moreover, $\Phi$ is called \emph{completely} positive if $\cI_R\otimes \Phi$ is positive for all $\cH_R$. We recall the following results which were proved in \cite{DJKR}.

\begin{thm}\label{thm:CB-norm-CPM} \cite{DJKR}
Let $\Phi:\bL(\cH_S) \rightarrow \bL(\cH_{S'})$ be completely positive. Then in~\eqref{def:CB-norm} we may restrict the second supremum 
to include only positive definite $Y_{RS}$, i.e.,
\bee
{\| \Phi \|}_{CB, q \rightarrow p} = \sup_{d_R} \, \sup_{Y_{RS}> 0} \,\frac{\| ({\cal I}_R \ot \Phi)(Y_{RS}) \|_{(t,p)}}{\| Y_{RS} \|_{(t,q)}}.
\eee
Moreover, if $p \le q$ then the supremum over $d_R$ is not required, i.e.,
\bee
{\| \Phi \|}_{CB, q \rightarrow p} =\|\Phi\|_{q\rightarrow p}= \sup_{Y_{S}> 0} \,\frac{\|  \Phi(Y_{S}) \|_{p}}{\| Y_{S} \|_{q}}.
\eee

\end{thm}

It is also shown in~\cite{DJKR} that the completely bounded norm is multiplicative for completely positive super-operators as in~\eqref{eq:CB-norm-multiplicative}. 

%******************************************************
\section{Main results}
Let $\{\Phi_t \, :\, \, t \geq 0\}$ be a semi-group of completely positive super-operators on $\bL(\cH_S)$
with generator $\mathcal L$. That is, 
for any $t\geq 0$, $\Phi_t:\bL(\cH_S)\rightarrow \bL(\cH_S)$ is completely positive,
and we have $\Phi_0=\cI_S$, and $\Phi_{s+t} = \Phi_s \Phi_t$. We let 
$$\mathcal L=-\frac{\dd}{\dd t} \Phi_t\Big|_{t=0}.$$
The super-operator $\mathcal L$ is called the \emph{Lindblad generator} of the semi-group.
Thus, for every $X_S$ and $t\geq 0$ we have
\bee
\frac{\dd}{\dd t} \Phi_t(X_S) = -\mathcal L \, \Phi_t(X_S)= -\Phi_t (\mathcal L X_S).
\eee
Equivalently, for every $t\geq 0$ and $X_S$ we have
\begin{align}\label{eq:Phi-exp}
\Phi_t(X_S) = e^{-t\mathcal{L}} (X_S).
\end{align}

We assume that the semi-group $\{\Phi_t\, :\, t\geq 0\}$ implements the Markov approximation for a quantum dynamics 
on $\cH_S$, in the Heisenberg representation. As a result, $\Phi_t$ for every $t$ is unital, meaning that $\Phi_t(I_S) = I_S$. This in particular implies that $\mathcal L (I_S)=0$.
The Schr\"odinger representation $\Phi_t^{*}$ is obtained by duality with respect to the Hilbert-Schmidt inner product, that is
\bee
\tau \big(( \Phi_t(X_S))^{*} \, \rho_S\big) = \tau \big( (X_S)^{*} \, \Phi_t^{*} (\rho_S)\big),
\eee
for all observables $X_S \in \bL(\cH_S)$ and all states $\rho_S \in {\cD}^{+}_{S}$.
The generator of the semigroup $\{\Phi_t^*:\, t\geq 0\}$ is $ \mathcal L^*$, the adjoint of $ \mathcal L$.  In the Schr\"odinger picture 
the quantum dynamics is trace-preserving, meaning that  $\tau \left( \Phi_t^{*}(\rho) \right) = \tau (\rho)$ for all $\rho$. Indeed, $\Phi_t$ is unital if and only if $\Phi_t^*$ is trace-preserving. 

In this paper we further assume that the semigroup is {\em reversible}, which means that $\Phi_t = \Phi_t^*$ for all $t$.
By the above discussion a reversible semigroup is both unital and trace-preserving. Moreover, its generator is
self-adjoint, that is
\be\label{l=l*}
\mathcal L = \mathcal L^* \,.
\ee

A completely positive super-operator $\Phi_t$ that is both unital and trace-preserving is a contraction under the Schatten $q$-norm,
for every $q \ge 1$,~\cite{PWPR}. That is, for every $q\geq 1$ we have
\begin{align}\label{eq:q-contraction}
\|\Phi_t\|_{CB, q\rightarrow q} = \|\Phi_t\|_{ q\rightarrow q}\leq 1,
\end{align}
where the equality $\|\Phi_t\|_{CB, q\rightarrow q} = \|\Phi_t\|_{ q\rightarrow q}$ holds by Theorem~\ref{thm:CB-norm-CPM}.
Then the following question arises. For a given $q$, what is the largest $p^*=p^*(t)$ such that 
$$\|\Phi_t\|_{CB, q\rightarrow p^*} \leq 1 \,,$$
for all $t \ge 0$?
Observe that $\| \cdot \|_p$ is a non-decreasing function of $p$, and so
by our assumption~\eqref{eq:q-contraction},  and the fact that $\{\Phi_t:\, t\geq 0\}$ forms a semigroup, 
it follows that $p^*(t)$ is a non-decreasing function of $t$ with $p^*(0)=q$.

\begin{defn}
Consider $q \ge 1$, and let $p=p(t)\geq q$ be defined for all $t \ge 0$. 
We say that the semigroup $\{\Phi_t: t\geq 0\}$ is completely bounded-$q$-hypercontractive (CB-$q$-HC) for $p(t)$ if for all $t \ge 0$ we have
\be\label{CB-hyper2}
{\| \Phi_t \|}_{CB, q \rightarrow p(t)} \le 1.
\ee
\end{defn} 

Next we define our notion of completely bounded log-Sobolev inequality. 

\begin{defn}
We say that the semigroup $\Phi_t$ with generator ${\mathcal L}$ satisfies the completely bounded (CB) log-Sobolev inequality at $q$ with constant $\alpha> 0$ if
\be\label{CB-logSob2}
 \tau(Y_{RS}^q \ln Y_{RS}^q) -  \tau_R \Big( \tau_S(Y_{RS}^q) \, \ln \tau_S(Y_{RS}^q)\Big) \le 
 {\alpha}{q^2}\,\tau \left( Y_{RS}^{q-1} \, ({\cal I}_R \otimes {\mathcal L}) (Y_{RS}) \right),
\ee
for all $\cH_R$ and all positive semidefinite $Y_{RS}\in \bL(\mathcal H_{RS})$. Here $\tau_R$ and $\tau_{S}$ denote the partial traces over 
$\cH_R$ and $\cH_S$ respectively.
\end{defn}

In the above definition, as usual, we extend the meaning of $\tau(Y\ln Y)$ to include positive semidefinite matrices,
by restricting to the  support of $Y$.

The main result of this paper is the equivalence of the above two definitions in the following sense.

\begin{thm}\label{thm:CB-HC-LS2}
Let $\{\Phi_t: t\geq 0\}$ be a semigroup of completely positive super-operators that satisfy~\eqref{eq:q-contraction}. Also 
consider $q \ge 1$ and let $p(t)\geq q$ 
(defined for $t \ge 0$) be a twice continuously differentiable increasing function with $q=p(0)$.
Then the following results hold:
\begin{enumerate}
\item[{\rm (i)}] If the semigroup is CB-$q$-HC for $p(t)$, then it satisfies the CB log-Sobolev inequality at $q$ with constant $\alpha= 1/p'(0)$. 

\item[{\rm (ii)}] If the semigroup satisfies the CB log-Sobolev inequality at $p(t)$ with constant $\alpha(t) = 1/p'(t)$ for all $t \ge 0$, then it is CB-$q$-hypercontractive for $p(t)$.
\end{enumerate}
\end{thm}

We also prove a CB version of the `Gross Lemma' \cite{Gross1} which relates the log-Sobolev inequalities at $q=2$ and $q > 2$.
This requires the additional assumption that the generator is self-adjoint.

\begin{thm}\label{thm:CB-LS-Gross}
Let $\{\Phi_t:\, t\geq 0\}$ be a semigroup of completely positive unital super-operators with self-adjoint generator $\mathcal L$ satisfying~\eqref{l=l*}.
Suppose $\Phi_t$ satisfies the CB log-Sobolev inequality~\eqref{CB-logSob2} at  $q=2$ with constant $\alpha> 0$.
Then $\Phi_t$ also satisfies \eqref{CB-logSob2} for all $q \ge 2$ with constant $\alpha (q-1)^{-1} > 0$.
\end{thm}

We note that unlike the usual log-Sobolev inequality, the CB log-Sobolev inequality in the non-commutative case satisfies the following tensorization property.

\begin{thm}\label{thm:tensorization}
Suppose that for all $i=1, \dots, k$ the semigroup of completely positive super-operators $\Phi_t^{(i)}:\bL(\cH_{S_i})\rightarrow \bL(\cH_{S_i})$ generated by $\mathcal L_i$  satisfies the CB log-Sobolev inequality at $q$ with constant $\alpha_i$. Then the semigroup $\Psi_t:\bL(\cH_{S_1\dots S_k})\rightarrow \bL(\cH_{S_1\dots S_k})$ generated by 
\bee
\mathcal L = \sum_i \hat{\mathcal L}_i,
\eee
where $\hat{\mathcal L}_i$ is obtained from $\mathcal L_i$ by tensoring with an appropriate identity super-operator,  satisfies the CB log-Sobolev inequality at $q$ with constant $\alpha = \max\{\alpha_1, \dots, \alpha_k\}$.

\end{thm}

In order to prove part (i)  of Theorem~\ref{thm:CB-HC-LS2} we will derive a formula for the derivative of the non-commutative $(q,p)$-norm at $p=q$. Since this result has independent interest we state it as a separate theorem.

\begin{thm}\label{thm:der-NC-norm}
Let $p(t)\geq 1$ be a twice continuously differentiable increasing function with $q=p(0)$ and $p'(0)=1/\alpha >  0$. Also let $X_{RS}(t)$ be a matrix-valued twice continuously differentiable function, where $X_{RS}(t)$ is positive definite in a neighborhood of $0$.
Let $Y = X_{RS}(0)$. Then
\begin{align}\label{eqn:der-NC-norm}
\frac{d}{d t} \big\| X_{RS}(t) \big\|_{(q,p(t))} \bigg|_{t=0} &= \frac{1}{\alpha q^2 \, \| Y \|_{q}^{q-1}} \,
\bigg[ \tau (Y^q \ln Y^q) - \tau_R \Big(\tau_S(Y^q) \ln \tau_S(Y^q)\Big) \nonumber \\
 &\hskip 2.9in + \alpha q^2 \tau (Y^{q-1} X_{RS}'(0)) \bigg].
\end{align}
\end{thm}

The main complication in proving the above theorems is that the 
definition of the $(q, p)$-norm involves an infimum or supremum (depending on whether $p\geq q$ or $q\geq p$) whose optimal point is not easy to compute. In the following section we derive some properties of the optimizer $\sigma_R$ in~\eqref{NC-norm-inf} and~\eqref{NC-norm-sup}, and in subsequent
sections we will use these properties to establish our results.

We finish this section with a few remarks about applications of our results. Observe that our notion of CB log-Sobolev inequality is stronger than the usual log-Sobolev inequality. This can be verified by taking the Hilbert space $\cH_R$ in~\eqref{CB-logSob2} to be trivial. As a result, the CB log-Sobolev constant is as large as the usual log-Sobolev constant. This fact can also be verified using the fact that the completely bounded norm is lower bounded by the usual operator norm. As a consequence of this observation, Theorems~\ref{thm:CB-HC-LS2} and~\ref{thm:tensorization}, and the results of~\cite{KT1}, we find that the CB log-Sobolev constant at $q=2$ is an upper bound for the mixing time of an \emph{arbitrarily large product of independent copies} of a semigroup defined  by a strongly regular generator. We emphasis that this statement (for a product of independent copies) has not been proven for the usual log-Sobolev constant, and is generally presumed to be false.  Another application of our work is in computing the `CB hypercontractivity ribbon' of~\cite{DB14} for certain bipartite density matrices.

%***********************************************
\section{Analysis of the $(q, p)$-norm}\label{sec:q-p-norm}
In this section we present some formulas for the derivative of the norm expression appearing in the
definitions~\eqref{NC-norm-inf} and~\eqref{NC-norm-sup}. We also
partially characterize the optimizer in the definition of the $(q, p)$-norm.

As in the statement of Theorem~\ref{thm:der-NC-norm}, let $p(t)$ be an increasing twice continuously
differentiable function with $q=p(0)$ and $p'(0)=1/\alpha> 0$. Also let $X_{RS}(t)$ be a matrix-valued 
twice continuously differentiable function, where $X_{RS}(t)$ is positive definite for $t\in [-\eta, \eta]$ for some $\eta>0$. 
For simplicity we sometimes denote $X_{RS}(t)$ and $p(t)$ by $X_{RS}$ and $p$, but keep in mind that they depend on $t$. 

Define
\be\label{def:s}
s(t) = \frac{1}{q} - \frac{1}{p(t)}
\ee
and for  $\sigma_R \in {\cD}^{+}_{R}$ define
\be\label{def:M,s}
M(t, \sigma_R) = (\sigma_R^{-s( t)/2} \ot I_S) X_{RS}( t) (\sigma_R^{-s( t)/2} \ot I_S).
\ee
Thus $M$ is positive definite for any $t \in[-\eta, \eta]$. Let
\be\label{def:F}
F(t, \sigma_R) = \| M(t, \sigma_R) \|_{p(t)}.
\ee
Note that $p(0)=q$ so $s(0)=0$, and $F(0, \sigma_R) = \| X_{RS}(0) \|_{q}$ does not depend on $\sigma_R$. 

Our first result establishes basic convexity and concavity properties of \eqref{def:F}
as a function of $\sigma_R$ for fixed $t$ (in the next section we will prove a more refined result
for the case $p \ge q$).

\begin{lemma}\label{lem:basic-F-conv-conc}
For fixed $t$ in $(-\eta,\eta)$, the function $\sigma_R \mapsto F(t, \sigma_R)^p$ is convex for $1 \le q \le p(t) \le 2 q$
and concave for $1 \le p(t) \le q$.
\end{lemma}

\begin{proof}
We are concerned with the function
\begin{align}\label{Hiai}
\sigma_R\mapsto F(t, \sigma_R)^p &= \big\|(\sigma_R^{-s/2} \ot I_S) X_{RS}( t) (\sigma_R^{-s/2} \ot I_S)\big\|_p^p\nonumber\\
&= \big\| X_{RS}( t)^{1/2} (\sigma_R \ot I_S)^{-s} X_{RS}(t)^{1/2}\big\|_p^p \nonumber \\
&= \tau \big( X_{RS}( t)^{1/2} (\sigma_R \ot I_S)^{-s} X_{RS}( t)^{1/2} \big)^p.
\end{align}
Hiai~\cite[Theorem 1.1]{Hiai} has proven that 
the map 
$$\xi \mapsto \tau \left( W \xi^{-s} W^* \right)^p,$$ 
on the set of positive definite matrices
is convex if $0 \le s \le 1$ and $1/2 \le p \le 1/s$, and is concave if $0 \le - s \le 1$ and $1/2 \le p \le - 1/s$. 
We apply this result to \eqref{Hiai} with $\xi = \sigma_R \ot I_S$, and use the 
definition \eqref{def:s} to relate $s,p,q$. Hiai's conditions for convexity are satisfied when
$1 \le q \le p(t) \le 2 q$, and the conditions for concavity are satisfied when $1 \le p(t) \le q$.

\end{proof}

Our next result presents some smoothness properties of $F$, and also formulas for its
derivative with respect to $t$.

\begin{lemma}\label{lem:second-der}
\begin{enumerate}
\item
[{\rm (a)}]The function ${\partial}^2 F/{\partial} t^2$ is continuous on 
$(-\eta,\eta) \times {\cD}_R^{+}$.
\item[{\rm (b)}]
The function $\sigma_R \rightarrow F(t, \sigma_R)$ is continuously differentiable
for all $\sigma_R \in {\cD}^{+}_{R}$ and $t\in(-\eta, \eta)$.
\item[{\rm (c)}]
For  all $\sigma_R \in {\cD}^{+}_{R}$ and $t\in(-\eta, \eta)$
\begin{align*}
\frac{\partial}{\partial t} F(t,\sigma_R) &= \frac{p'(t)\, F(t,\sigma_R)}{p^2 \, \tau (M^p)} \, \bigg[
 - \tau (M^p) \,\ln \tau (M^p)+  \tau (M^p \ln M^p)   -  \tau_R \left( \tau_S(M^p) \ln \sigma_R \right) \nonumber\\
& \hskip 1.3in+  \frac{p^2}{p'(t)} \, \tau \left( M^{p-1} (\sigma_R^{-s( t)/2} \ot I_S) X_{RS}'( t) (\sigma_R^{-s( t)/2} \ot I_S)  \right) \bigg].
\end{align*}
In particular, we have
\begin{align}\label{eq:deriv-t=0}
\frac{\partial}{\partial t} F(t,\sigma_R) \bigg|_{t=0} 
&= \frac{1}{\alpha q^2 \, \| Y \|_q^{q-1}} \,\bigg[
 - \tau (Y^q) \,\ln \tau (Y^q) +  \tau (Y^q \ln Y^q)   \nonumber \\
 & \hskip 1.1in -  \tau_R \left( \tau_S(Y^q) (\ln \sigma_R)\right)
+  \alpha q^2 \, \tau \left( Y^{q-1} \, X_{RS}'( 0)  \right) \bigg],
\end{align}
where $Y = X_{RS}(0)$ and $\alpha=1/p'(0)$.
\end{enumerate}
\end{lemma}

Our main tool in the proof of this lemma is the contour integral representation of $M^p$; the rest is a straightforward calculation, so we leave the proof for Appendix~\ref{app:proof-lem-second-der}.

\medskip
Next we will use these basic results about derivatives to provide estimates for $F(t, \sigma_R)$
in a neighborhood of $t=0$. Note first that since $p(0)=q$ we have
\bee
F(0, \sigma_R) = \| M(0, \sigma_R) \|_q = \| X_{RS}(0) \|_q = \| Y \|_q,
\eee
where as before $ Y = Y_{RS}=X_{RS}(0)$. 
We define the normalized reduced density matrix of $Y_{RS}^q$ by
\be\label{eq:def-gamma}
\gamma_R = \frac{1}{\tau (Y^q)} \, \tau_S(Y^q),
\ee
Note that $Y$ is positive definite and $\tau_R (\gamma_R) = 1$, so $\gamma_R \in {\cD}^{+}_{R}$.
Let us also define $G(\sigma_R)$ to be the factor in braces on the right hand side of~\eqref{eq:deriv-t=0}, that is,
\begin{align}\label{def:G}
G(\sigma_R) =& - \tau (Y^q) \,\ln \tau (Y^q) +  \tau (Y^q \ln Y^q)  \nonumber\\
&\,-  \tau_R \left( \tau_S(Y^q) (\ln \sigma_R)\right)
+  \alpha q^2 \, \tau \left( Y^{q-1} \, X_{RS}'( 0) \right).
\end{align}

\begin{lemma}\label{lem:F-lin-approx}
There is $\kappa > 0$ and $K < \infty$, such that for all 
$t\in [-\eta/2, \eta/2]$ and $\sigma_R\in \mathcal{B}_\kappa(\gamma_R)$,
\be\label{eq:lin-approx}
\bigg|F(t, \sigma_R) - \| Y \|_{q} - t \frac{G(\sigma_R)}{\alpha q^2 \, \| Y \|_q^{q-1}} \bigg| \le K \, t^2.
\ee
\end{lemma}

\begin{proof}
Let $t\in [-\eta/2, \eta/2]$, and recall the definitions \eqref{def:ball} and \eqref{eq:def-gamma}.
Since $\cD_{R}^{+}$ is open, there is $\kappa>0$ such that 
\begin{align}\label{eq:def-kappa}
\mathcal{B}_\kappa(\gamma_R) \subset \cD_{R}^{+}.
\end{align}
Since $\mathcal{B}_\kappa(\gamma_R)$ is closed and bounded, it is a compact subset of $\cD_{R}^{+}$.
Furthermore by Lemma~\ref{lem:second-der}, $\partial^2 F/\partial t^2$ is continuous on $(-\eta, \eta)\times \cD_R^+$.
Hence there is $K < \infty$ such that 
\begin{align}\label{eq:def-K}
-2K\leq    \frac{\partial^2 F}{\partial t^2} (t, \sigma) \leq 2K,
\end{align}
for all $t\in [-\eta/2, \eta/2]$, and all $\sigma_R\in \mathcal{B}_\kappa(\gamma_R)$.
Therefore, for any $t\in [-\eta/2, \eta/2]$ and $\sigma_R\in \mathcal{B}_\kappa(\gamma_R)$ we have 
\begin{align*}
\bigg|F(t, \sigma_R) - F(0, \sigma_R) - t\frac{\partial F}{\partial u}(u, \sigma_R)  \Big|_{u=0} \, \bigg| = 
\bigg|  \int_{0}^t  (t-u)\frac{\partial^2}{\partial u^2} F(u, \sigma_R) \dd u  \bigg|
 \leq Kt^2.
\end{align*}
Noting that $F(0, \sigma_R)=\|Y\|_q$ and using the definition of $G( \sigma_R)$ we find that
\bee
\bigg|F(t, \sigma_R) - \| Y \|_{q} - t \frac{G(\sigma_R)}{\alpha q^2 \, \| Y \|_q^{q-1}} \bigg| \le K \, t^2,
\eee
for all $t\in [-\eta/2, \eta/2]$, and $\sigma_R\in \mathcal{B}_\kappa(\gamma_R)$.

\end{proof}

Returning to the formula \eqref{def:G} and using the definition \eqref{eq:def-gamma}, we observe
that for any $\sigma_R \in {\cD}^{+}_{R}$,
\begin{align}\label{G-rel-ent}
G( \sigma_R) =& G(\gamma_R) + 
\tau(Y^q)\,\tau_R (\gamma_R \, \ln \gamma_R)  - \tau(Y^q)\,\tau_R (\gamma_R \, \ln \sigma_R)  \nonumber \\
=& G(\gamma_R) +  \tau(Y^q) \,S\left(d_R^{-1} \gamma_R \big\| d_R^{-1} \sigma_R \right),
\end{align}
where $S(\cdot \| \cdot)$ is the relative entropy between density matrices $\gamma_R/d_R$ and $\sigma_R/d_R$ defined by
\be\label{def:rel-ent}
S\left(d_R^{-1} \gamma_R \big\| d_R^{-1} \sigma_R \right) = \tr\big(d_R^{-1} \gamma_R\big (\ln(d_R^{-1} \gamma_R) - \ln(d_R^{-1} \sigma_R))\big) = \tau\big(\gamma_R(\ln \gamma_R - \ln\sigma_R)\big).
\ee
%%%%%

Our final lemma in this section localizes the optimizer in the $(q,p)$ norm for small $t$.

%%%%%%
\begin{lemma}\label{lem:min-max-exist}
For any $0 < \epsilon \le \kappa$, where $\kappa$ is the parameter described in Lemma \ref{lem:F-lin-approx}, 
there is $\delta > 0$ such that 
for all $t\in [-\delta, \delta]$ there is
$\widetilde\sigma_R( t) \in {\cD}_R^{+}$ satisfying
$
\| X_{RS}( t) \|_{(q,p)} = F(t, \widetilde\sigma_R( t))
$
and
\bee
\| \gamma_R - \widetilde\sigma_R( t) \|_1 \le \epsilon.
\eee
\end{lemma}

\begin{proof}
Given $\epsilon \le \kappa$, where $\kappa$ was defined in~\eqref{eq:def-kappa}, we choose
$\delta'>0$ to satisfy
\bee
\delta' < \min\left\{ \frac{\eta}{2},  \frac{\epsilon^2\, \tau(Y^q)}{4 K\, \alpha q^2 \, \| Y \|_q^{q-1}}  \right\} ,
\eee
where $K$ is defined by~\eqref{eq:def-K}. We have
\bee
\cB_\epsilon(\gamma_R) \subset
\mathcal{B}_\kappa(\gamma_R) \subset \cD_{R}^{+},
\eee
and so the boundary of $\cB_\epsilon(\gamma_R)$ is contained in $\cD_{R}^{+}$.
Suppose that $\sigma_R$ is on the boundary of $\cB_\epsilon(\gamma_R)$, so that
\be\label{def:boundary}
\| \gamma_R - \sigma_R \|_1 = \epsilon.
\ee
Pinkser's inequality \cite{Pinsker} implies that
\bee
S\left(d_R^{-1} \gamma_R \big\| d_r^{-1} \sigma_R \right) \ge \frac{1}{2} \| \gamma_R - \sigma_R \|_1^2 = \frac{\epsilon^2}{2},
\eee
where $S$ is the relative entropy defined in \eqref{def:rel-ent}.
Thus from \eqref{G-rel-ent} we deduce
\be\label{G-ineq1}
G(\sigma_R) \ge G(\gamma_R) + \frac{\epsilon^2\,\tau(Y^q)}{2}.
\ee

We consider first the case where $t \ge 0$.
From~\eqref{eq:lin-approx} we deduce that
\bee
F(t, \sigma_R) &\ge& \| Y \|_q + t \frac{G(\sigma_R)}{\alpha q^2 \, \| Y \|_q^{q-1}}   - K \, t^2  \\
& \ge & \| Y \|_q + t \frac{G(\gamma_R)}{\alpha q^2 \, \| Y \|_q^{q-1}}  + t \frac{\epsilon^2\,\tau(Y^q)}{2\, \alpha q^2 \, \| Y \|_q^{q-1}} - K \, t^2.
\eee
Our choice of $\delta'$ implies that for all $0\leq t \le \delta'$ we have
\bee
t \frac{\epsilon^2\,\tau(Y^q)}{2\, \alpha q^2 \, \| Y \|_q^{q-1}} - K \, t^2 > K \, t^2,
\eee
and thus
\be\label{F-ineq1}
F(t, \sigma_R) > \| Y \|_q + t \frac{G(\gamma_R)}{\alpha q^2 \, \| Y \|_q^{q-1}}  + K \, t^2.
\ee
Furthermore, from~\eqref{eq:lin-approx} we also deduce that
\be\label{F-ineq2}
F(t, \gamma_R) \le \| Y \|_q + t \frac{G(\gamma_R)}{\alpha q^2 \, \| Y \|_q^{q-1}}   + K \, t^2.
\ee
Combining \eqref{F-ineq1} and \eqref{F-ineq2} we find that 
\bee
F(t, \gamma_R) < F(t, \sigma_R).
\eee
Since this inequality holds for all $\sigma_R$ on the boundary of $\cB_\epsilon(\gamma_R)$, we conclude that
for all $0\leq t \le \delta'$ the function $\sigma_R\mapsto F(t, \sigma_R)$ has a local minimum $\widetilde\sigma_R(t)$ in the interior of $\cB_\epsilon(\gamma_R)$.
We now choose $0 < \delta_{+} \le \delta'$ so that $q \le p(t) \le 2 q$ for all $0 \le t \le \delta_{+}$ (the existence of $\delta_{+} > 0$
is guaranteed by our assumptions that $p(0)=q \ge 1$ and that $p(t)$ is increasing and differentiable).
Applying Lemma \ref{lem:basic-F-conv-conc} we conclude that the local minimum of the convex function
$\sigma_R\mapsto F(t, \sigma_R)^{p(t)}$ in the interior of $\cB_\epsilon(\gamma_R)$ is in fact a global minimum
for all $0 \le t \le \delta_{+}$. Since $F(t, \sigma_R)$ and $F(t, \sigma_R)^p$ share the same minimum $\widetilde\sigma_R(t)\in \cB_\epsilon(\gamma_R)$,
we conclude that $\| X_{RS}( t) \|_{(q,p)} = F(t, \widetilde\sigma_R( t))$ and
\bee
\| \gamma_R - \widetilde\sigma_R( t) \|_1 \le \epsilon.
\eee

Turning to the case $t \le 0$ we use~\eqref{eq:lin-approx} and \eqref{G-ineq1} to deduce that
\bee
F(t, \sigma_R) &\le& \| Y \|_q + t \frac{G(\sigma_R)}{\alpha q^2 \, \| Y \|_q^{q-1}}   + K \, t^2  \\
& \le & \| Y \|_q + t \frac{G(\gamma_R)}{\alpha q^2 \, \| Y \|_q^{q-1}}  + t \frac{\epsilon^2\,\tau(Y^q)}{2\, \alpha q^2 \, \| Y \|_q^{q-1}} + K \, t^2.
\eee
Again using the definition of $\delta'$ and noting that $t$ is negative, we have
\bee
t \frac{\epsilon^2\,\tau(Y^q)}{2\, \alpha q^2 \, \| Y \|_q^{q-1}} + K \, t^2 < - K \, t^2,
\eee
and thus
\be\label{F-ineq3}
F(t, \sigma_R) < \| Y \|_q + t \frac{G(\gamma_R)}{\alpha q^2 \, \| Y \|_q^{q-1}}  - K \, t^2.
\ee
Combining this time with the lower bound for $F(t, \gamma_R)$ obtained from~\eqref{eq:lin-approx} we deduce that
\bee
F(t, \gamma_R) > F(t, \sigma_R),
\eee
for all $\sigma_R$ on the boundary of $\cB_\epsilon(\gamma_R)$. Thus we conclude that
for all $-\delta' \leq t \le 0$ the function $\sigma_R\mapsto F(t, \sigma_R)$ has a local maximum in the interior of $\cB_\epsilon(\gamma_R)$.
We now choose $0 < \delta_{-} \le \delta'$ so that $1 \le p(t) \le q$ for all $- \delta_{-} \le t \le 0$.
Applying Lemma \ref{lem:basic-F-conv-conc} we conclude that the local maximum of the concave function
$\sigma_R\mapsto F(t, \sigma_R)^{p(t)}$ in the interior of $\cB_\epsilon(\gamma_R)$ is in fact a global maximum
for all $- \delta_{-} \le t \le 0$. 

Finally we take $\delta = \min \{\delta_{+}, \delta_{-} \}$ and deduce that for all $t\in [-\delta, \delta]$ there is
$\widetilde\sigma_R( t) \in {\cD}_R^{+}$ satisfying
$
\| X_{RS}( t) \|_{(q,p)} = F(t, \widetilde\sigma_R( t))
$
and
\bee
\| \gamma_R - \widetilde\sigma_R( t) \|_1 \le \epsilon.
\eee

\end{proof}

%%%%%
\subsection{Restriction to $p> q$}
We now restrict our attention to $t> 0$, in which case $p > q$.
We will prove a refined characterization of the optimal 
$\sigma_R$ which holds for all $p>q$ (and not just for small $t> 0$).

\begin{lemma}\label{lem1}
For a fixed $t\in (0, \eta)$, for which $p>q$, the function
\be\label{eq:F-fixed-t}
\sigma_R\mapsto F(t, \sigma_R),
\ee
is strictly convex, and there is a unique $\hat\sigma_R \in {\cD}^{+}_R$ such that
\be\label{eq:q-p-norm-F}
F(t, \hat\sigma_R) = \|X_{RS}(t)\|_{(q, p)}.
\ee
Moreover, the optimizer $\hat\sigma_R$ in~\eqref{eq:q-p-norm-F} satisfies
\be\label{lem1:eqn2}
\hat\sigma_R = \frac{1}{F(t, \hat\sigma_R)^p} \, \tau_{S} \left[ \left( (\hat\sigma_R^{-s/2} \ot I_S ) X_{RS}(t)  (\hat\sigma_R^{-s/2} \ot I_S ) \right)^p \right].
\ee
\end{lemma}

\medskip
\begin{proof}
We borrow ideas from the proof of Lemma~20 of~\cite{HT14} in order to prove this lemma.
By the unitary invariance of the $p$-norm we can rewrite the function $F$ as
\bee
F(t, \sigma_R) =  \| X(t)^{1/2} (\sigma^{-s} \ot I_S ) X(t)^{1/2} \|_p.
\eee
Since $p>q$ we have $s \in (0,1]$ (see \eqref{def:s}). Then the map
\be\label{eq:sigma-map-s}
\sigma \mapsto \sigma^{-s},
\ee
is operator convex~\cite{Bhatia}, and thus for any $\lambda\in [0,1]$ and $\sigma_R, \xi_R\in \cD_R^+$ we have
$$(\lambda\sigma_R + (1-\lambda)\xi_R)^{-s} \leq \lambda\sigma_R^{-s} + (1-\lambda)\xi_R^{-s}.$$
Next, the monotonicity of the map $\zeta \rightarrow X^{1/2} \zeta X^{1/2}$ and of the $p$-norm imply
\bee
F(t, \lambda \sigma + (1 - \lambda) \xi) \le \big\| \lambda X^{1/2} (\sigma^{-s} \ot I_S ) X^{1/2} + (1 - \lambda) X^{1/2} (\xi^{-s} \ot I_S ) X^{1/2} \big\|_p.
\eee
For all $p \ge 1$ the Schatten $p$-norm is uniformly convex \cite{Ball-Carlen-Lieb}, and thus also strictly convex. Therefore
\begin{align*}
&\big\| \lambda X^{1/2} (\sigma^{-s} \ot I_S ) X^{1/2} + (1 - \lambda) X^{1/2} (\xi^{-s} \ot I_S ) X^{1/2} \big\|_p \\
& \hskip1in \le 
\lambda \big\| X^{1/2} (\sigma^{-s} \ot I_R ) X^{1/2}  \big\|_p + (1 - \lambda) \big\| X^{1/2} (\xi^{-s} \ot I_S ) X^{1/2} \big\|_p \\
& \hskip1in =
\lambda F(t, \sigma) + (1-\lambda) F(t, \xi),
\end{align*}
with equality if and only if
\bee
X^{1/2} (\sigma^{-s} \ot I_S ) X^{1/2} = c \, X^{1/2} (\xi^{-s} \ot I_S ) X^{1/2},
\eee
for some $c \in \mathbb{R}$. Since $X$ is positive definite (and therefore invertible), 
the equality condition is equivalent to $\sigma^{-s} = c \, \xi^{-s}$ which by the normalization $\tau(\sigma)=\tau(\xi)=1$ gives $\sigma=\xi$.
We conclude that
\bee
F(t, \lambda \sigma + (1 - \lambda) \xi) \le \lambda F(t, \sigma) + (1-\lambda) F(t, \xi),
\eee
with equality if and only if $\sigma = \xi$. Therefore, the function $F(t, \sigma_R)$ is strictly convex in $\sigma_R$.

Now we will show that the infimum in~\eqref{eq:q-p-norm-F} is achieved. 
We argue by contradiction, so suppose that the infimum is not achieved in $\cD_R^+$.
Then there must exist a non-convergent sequence $\{\xi_n: \, n\geq 1\}\subset \cD_R^+$ such that
\be\label{eq:lim-F}
\lim_{n\rightarrow \infty} F(t, \xi_n) = \|X_{RS}(t)\|_{(q, p)}.
\ee
The closure of $\cD_R^+$ is compact in $\bL(\cH_R)$, and thus the sequence $\{\xi_n:\, n\geq 1 \}$ has a limit point
in $\bL(\cH_R)$. By assumption there is no limit point in $\cD_R^+$, thus the limit point
belongs to the boundary $\partial \cD_R^+$. So there is a subsequence $\{ \xi_{n_j}:\, j\geq 1 \}$
which approaches $\partial \cD_R^+$ as $j \rightarrow \infty$.
Since $\partial \cD_R^+$ consists of singular matrices, and $s>0$, we obtain
\begin{align}\label{eq:lim-norm-xi}
\lim_{j \rightarrow \infty} \|\xi_{n_j}^{-s}\|_p = \infty.
\end{align}
Now since $X_{RS}(t)$ is positive definite (and thus invertible) we have
\begin{align*}
F(t, \xi_{n_j})& = \| (\xi_{n_j}^{-s/2} \ot I_S) X_{RS} (\xi_{n_j}^{-s/2} \ot I_S)  \|_p\\
& = \| X_{RS}^{1/2}(\xi_{n_j}^{-s} \ot I_S) X_{RS}^{1/2}  \|_p\\
& \geq \|\ \xi_{n_j}^{-s} \ot I_S \|_p \,\|X_{RS}^{-1/2}\|_{\infty}^{-2} \\
& = \|\ \xi_{n_j}^{-s}  \|_p \,\|X_{RS}^{-1/2}\|_{\infty}^{-2},
\end{align*}
which implies that $F(t, \xi_{n_j}) \rightarrow \infty$ as $j \rightarrow \infty$. This contradicts
our assumption~\eqref{eq:lim-F}.  So we conclude that the infimum in~\eqref{eq:q-p-norm-F} is achieved
in $\cD_R^+$. Moreover, by the strict convexity proved above, the infimum is achieved at a unique point which we call $\hat \sigma_R(t)$.

Next we show that $\hat\sigma_R$ satisfies equation~\eqref{lem1:eqn2}.
For this purpose we recall Lemma \ref{lem:second-der}({b}), where we showed that
$\xi_R\mapsto F(t, \xi_R)$ is a continuously differentiable function in ${\cD}^{+}_R$. Since the function
has a minimum at $\hat\sigma_R$, its 
derivative must vanish at $\xi_R= \hat\sigma_R$. To compute the derivative, let $\varrho$ be a traceless hermitian matrix, and define
\be\label{def:xi-x}
\xi(x) = \hat \sigma + x \varrho.
\ee
Then $\xi(x) \in {\cD}^{+}_R$ for all sufficiently small $|x|$.
Let
\bee
B(x) = X^{1/2} ( \xi(x)^{-s/2} \ot I_S ).
\eee
Then we have
\bee
F(t, \xi(x))^p = \tau \big((B^* B)^p\big),
\eee
and therefore,
\begin{align*}
\frac{\dd}{\dd x} F(t, \xi(x))^p &= p \, \tau \left((B^* B)^{p-1} \Big[\frac{\dd B^*}{\dd x} B + B^* \frac{\dd B}{\dd x} \Big]\right).
\end{align*}
Define
\be
\psi(x) = \xi(x)^{s/2} \frac{\dd}{\dd x} \xi(x)^{-s/2}.
\ee
Then we have
\be
\frac{\dd B}{\dd x} = B ( \psi \ot I_S).
\ee
Therefore,
\begin{align*}
\frac{\dd}{\dd x} F(t, \xi(x))^p &= p \,\tau\left((B^* B)^{p-1}\big[ (\psi^* \ot I_S ) B^* B + B^* B (\psi \ot I_S) \big]\right)\\
&=p \,\tau\left((B^* B)^{p}\big[ (\psi^* +\psi) \ot I_S \big]\right).
\end{align*}
Let 
\bee
N_R = \tau_S \big((B(0)^* B(0))^{p}\big) = \tau_S\Big[ \left( (\hat\sigma_R^{-s/2} \ot I_S ) X_{RS}(t)  (\hat\sigma_R^{-s/2} \ot I_S ) \right)^p   \Big].
\eee
Then we have
\be\label{deriv}
\frac{\dd}{\dd x} F(t, \xi(x))^p\Big|_{x=0} = p \, \tau_R 
\big(N_R (\psi^*(0)  + \psi(0))\big) = p \, \tau_R \big( \hat\sigma^{-1/2} N_R \hat\sigma^{-1/2} \Gamma(\varrho)\big),
\ee
where
\bee
\Gamma(\varrho) = \hat\sigma^{s/2+1/2} \frac{\dd}{\dd x} (\xi^{-s/2}) \hat\sigma^{1/2}\Big|_{x=0} + \hat\sigma^{1/2} \frac{\dd}{\dd x} (\xi^{-s/2}) \hat\sigma^{s/2 + 1/2}\Big|_{x=0}. 
\eee

We claim that $\varrho\mapsto \Gamma(\varrho)$ maps the subspace of traceless Hermitian matrices
into itself, and is onto. To see this,
we first extend the definition of $\Gamma$ to a linear operator
$\hat{\Gamma}$ on the space of all Hermitian matrices, by extending~\eqref{def:xi-x}
to allow general Hermitian matrices $\rho$.
We claim that $\hat{\Gamma}$ is surjective. To see this,
first note that the map $\zeta \rightarrow \zeta^{-s/2}$ is one-to-one on positive definite matrices, and hence its derivative
$$\varrho\mapsto \frac{\dd}{\dd x} (\hat\sigma + x \varrho)^{-s/2}\Big|_{x=0},$$ 
is onto. The map
$$\zeta \rightarrow \hat \sigma^{s/2 + 1/2} \zeta \hat\sigma^{1/2} + \hat\sigma^{1/2} \zeta \hat\sigma^{s/2 + 1/2},$$ 
is also onto. As a result their composition which is
$\hat{\Gamma}$ is onto. Now we note that
\bee
\tau (\hat{\Gamma} (\varrho) )= 2 \tau\Big( \hat\sigma^{1+s/2}  \frac{\dd}{\dd x} \xi^{-s/2}  \Big)\Big|_{x=0}= 
- s \tau\Big(\frac{\dd}{\dd x} \xi\Big)\Big|_{x=0} = - s \tau(\varrho).
\eee
Therefore,  $\hat{\Gamma}$ maps the subspace of traceless Hermitian matrices into itself, and is onto.
Thus its restriction to the traceless Hermitian matrices, namely $\Gamma$, is also onto.

Returning to~\eqref{deriv}, we conclude that for any traceless Hermitian matrix $\zeta$ we have 
$$\tau\Big(\hat\sigma^{-1/2} N_R \hat\sigma^{-1/2}\zeta\Big)=0.$$
Therefore $\hat\sigma^{-1/2} N_R \hat\sigma^{-1/2}$ is a multiple of the identity matrix. Thus $\hat\sigma$ is proportional to $N_R$, 
and since $\tau(\hat\sigma)=1$ we must have~\eqref{lem1:eqn2}.

\end{proof}
%%%%%

According to Lemma \ref{lem1}, for any $t>0$ there is a unique $\hat\sigma_R(t) \in \cD_R^+$ such that 
$\|X_{RS}(t)\|_{(q, p)} = F(t, \hat\sigma_R(t))$. Moreover, from the results of Lemma~\ref{lem:min-max-exist} 
we can conclude that for sufficiently small $t>0$, $\hat\sigma_R(t)$ is close to $\gamma_R$.  

We will use the following continuity result in the next section when we apply these lemmas to prove
our main theorem. For $t \ge 0$ we define
\be\label{def:varphi}
\varphi(t) = \|X_{RS}(t)\|_{(q, p(t))} = F(t, \hat\sigma_R(t)).
\ee

\begin{lemma}\label{lem:continuous}
$\varphi(t)$ is continuous on $[0, \eta)$.
\end{lemma}

\begin{proof}
We first prove continuity at $t=0$. Recalling Lemma~\ref{lem:F-lin-approx}, there is $\kappa > 0$ and
$K < \infty$ such that for all $\sigma_R\in \cB_\kappa(\gamma_R)$ and $t\in [0, \eta/2)$ we have 
\begin{align}\label{eq:f-conf}
\big|F(t, \sigma_R) - \| Y \|_{q} \big| \le t \frac{G(\sigma_R)}{\alpha q^2 \, \| Y \|_q^{q-1}} +K \, t^2.
\end{align}
On the other hand,  for sufficiently small $t>0$, Lemma~\ref{lem:min-max-exist} implies that
the optimizer $\hat\sigma_R(t)$  is in $\cB_\kappa(\gamma_R)$. 
Thus, noting that $\varphi(0) = \| Y \|_{q}$, for sufficiently small $t > 0$, we deduce
\be\label{cont-varphi}
\big|\varphi(t) - \varphi(0) \big| \le t \frac{G(\hat{\sigma}_R)}{\alpha q^2 \, \| Y \|_q^{q-1}} +K \, t^2.
\ee
Recalling \eqref{def:G} we note that the function $\sigma_R \mapsto G(\sigma_R)$ is continuous,
and thus uniformly bounded on $\cB_\kappa(\gamma_R)$. Therefore the bound~\eqref{cont-varphi} implies continuity of $\varphi(t)$ at $t=0$.

Now consider any $t_0 \in (0, \eta)$.
We will prove continuity of $\varphi$ at $t_0$. Let $0<a<t_0< b<\eta$ be arbitrary. For $t\in [a, b]$, we have
\begin{align*}
\varphi(t)& =F(t, \hat\sigma_R(t)) \\
&= \|(\hat\sigma_R(t)^{-s(t)/2}\ot I_S) X_{RS}(t) (\hat\sigma_R(t)^{-s(t)/2}\ot I_S)\|_{p}\\
&\geq \|\hat\sigma_R(t)^{-s(t)}\|_p \, \|X_{RS}(t)^{-1/2}\|^{-2}_\infty \\
& \geq  d_R^{-1}\lambda_{\min}(\hat\sigma_R(t))^{-s(t)} \, \|X_{RS}(t)^{-1/2}\|^{-2}_\infty,
\end{align*}
where $\lambda_{\min}(\hat\sigma_R(t))$ is the minimum eigenvalue of $\hat\sigma_R(t)$,
and $d_R = \dim\mathcal H_R$. On the other hand,
$$\varphi(t)=\inf_{\sigma_R} F(t, \sigma_R) \leq F(t, I_R) = \|X_{RS}(t)\|_p\leq \|X_{RS}(t)\|_\infty.$$
Putting these together we conclude that
$$\lambda_{\min}(\hat\sigma_R(t))^{-s(a)}\leq \lambda_{\min}(\hat\sigma_R(t))^{-s(t)}\leq d_R\|X_{RS}(t)^{-1/2}\|^{2}_\infty\|X_{RS}(t)\|_\infty,$$
where we use the fact that $s(t)$ is increasing in $t$, and that $\hat\sigma_R(t)\in \cD_R^+$ which gives $\lambda_{\min}(\hat\sigma_R(t))\leq 1$.
Now we note that $X_{RS}(t)$ is invertible and continuous. So there is $C>0$ such that for all $t\in [a, b]$ we have
$$d_R\|X_{RS}(t)^{-1/2}\|^{2}_\infty\|X_{RS}(t)\|_\infty \leq C.$$
Therefore, $\{\hat\sigma_R(t):\, t\in [a, b]\}\subseteq \Lambda$ where
$$\Lambda= \{\sigma_R\in \cD_R^+:\,   \lambda_{\min}(\sigma_R)\geq C^{-1/s(a)}  \}.$$

The function $F(t, \sigma_R)$ restricted to the compact set $[a, b]\times \Lambda$ is continuous, and therefore
also uniformly continuous. Hence, for every $\epsilon>0$ there is $\delta>0$ such that for every 
$t, t'\in [a, b]$ with $| t - t' | < \delta$ and $\sigma_R\in \Lambda$ we have 
$$|F(t, \sigma_R) - F(t', \sigma_R)|\leq \epsilon.$$ This implies
\bee
\varphi(t) = F(t, \sigma_R(t) ) \le
F(t, \sigma_R(t') ) \le F(t', \sigma_R(t')) +\epsilon = 
\varphi(t') +\epsilon.
\eee
We similarly have $\varphi(t')\leq \varphi(t)+\epsilon$. As a result, 
$$|\varphi(t)-\varphi(t')|\leq \epsilon,$$
for all $| t - t' | < \delta$.
Therefore, $\varphi(t)$ is continuous in $[a, b]$, and in particular at $t=t_0$.

\end{proof}

%**************************************************************
\section{Proof of Theorem \ref{thm:der-NC-norm}}

We now have all the tools required to prove Theorem~\ref{thm:der-NC-norm}. 
By assumptions $\alpha = p'(0)^{-1}$ is positive and finite.
Using the definitions~\eqref{eq:def-gamma} and~\eqref{def:G} we find
\bee
G(\gamma_R) = \tau (Y^q \ln Y^q) - \tau_R \Big(\tau_S(Y^q) \ln \tau_S(Y^q)\Big) + \alpha q^2 \, \tau (Y^{q-1} X_{RS}'(0)),
\eee
which is the expression inside the braces on the right side of~\eqref{eqn:der-NC-norm}.
We define
\bee
\Delta(t) = 
\frac{1}{t}\Big(\|X_{RS}(t)\|_{(q, p(t))} - \| X_{RS}(0) \|_{(q, p(0))}\Big) - \frac{G(\gamma_R)}{\alpha \, q^2 \, \| Y \|_q^{q-1}}.
\eee
Thus our goal is to prove that $\Delta(t) \rightarrow 0$ as $t \rightarrow 0$. 

Let $0<\epsilon$ be such that
\be\label{def:epsilon}
\epsilon <  \min \{ \kappa, \eta, \frac{\lambda_{\min}(\gamma_R)}{2 d_R} \},
\ee
where $\kappa$ is the parameter described in Lemma \ref{lem:F-lin-approx},
$\lambda_{\min}(\gamma_R)$ is the minimum eigenvalue of $\gamma_R$ and as before
$d_R = \dim (\mathcal H_R)$.
According to Lemma~\ref{lem:min-max-exist}, there is $\delta > 0$ sufficiently small such that for every $0 < t < \delta$ there is an optimizer $\widetilde \sigma_R(t)$ such that
\bee
\| \widetilde\sigma_R(t) - \gamma_R \|_1 \le \epsilon \le \kappa,
\eee
and 
\bee
\|X_{RS}(t)\|_{(q, p(t))}=F(t,\widetilde{\sigma}_R(t)) .
\eee
Then we have
\be\label{Delta-expand}
\Delta(t) &=&
\frac{1}{t}\Big(F(t,\widetilde{\sigma}_R(t)) - \| Y \|_q \Big) - \frac{G(\gamma_R)}{\alpha \, q^2 \, \| Y \|_q^{q-1}} \nonumber \\
&=& \frac{1}{t}\Big(F(t,\widetilde{\sigma}_R(t)) - \| Y \|_q  - t \, \frac{G(\widetilde{\sigma}_R(t))}{\alpha \, q^2 \, \| Y \|_q^{q-1}} \Big) 
+ \frac{G(\widetilde{\sigma}_R(t)) - G(\gamma_R)}{\alpha \, q^2 \, \| Y \|_q^{q-1}}.
\ee
Since $\widetilde\sigma_R(t) \in \mathcal{B}_\kappa(\gamma_R)$, Lemma~\ref{lem:F-lin-approx}
implies that
\be\label{Delta-ineq1}
\bigg| F(t,\widetilde{\sigma}_R(t)) - \| Y \|_q  - t \, \frac{G(\widetilde{\sigma}_R(t))}{\alpha \, q^2 \, \| Y \|_q^{q-1}} \bigg| \le K \, t^2.
\ee
Furthermore, from~\eqref{G-rel-ent} and using Lemma~\ref{lem:rel-ent-Lip} in Appendix~\ref{app:rel-ent-Lip} we obtain
\be\label{Delta-ineq2}
\big|G(\widetilde\sigma_R(t)) - G(\gamma_R)\big| &= &
S\left(d_R^{-1} \gamma_R \bigg\| d_R^{-1} \widetilde\sigma_R(t) \right) \nonumber \\
&\le &  \frac{2 d_R}{\lambda_{\min}(\gamma_R)} \, \big\| \gamma_R - \widetilde\sigma_R(t) \big\|_1 \nonumber  \\
& \le&  \frac{2 d_R}{\lambda_{\min}(\gamma_R)} \, \epsilon.
\ee
Using \eqref{Delta-ineq1} and \eqref{Delta-ineq2} in \eqref{Delta-expand} we obtain the bound
\bee
| \Delta(t) | \le K \, t + \frac{2 d_R}{\lambda_{\min}(\gamma_R) \, \alpha \, q^2 \, \| Y \|_q^{q-1}} \, \epsilon,
\eee
for all $\epsilon$ satisfying \eqref{def:epsilon}, and all $0 < t < \delta$. 
Therefore
\bee
\limsup_{t \rightarrow 0} | \Delta(t) | \le \frac{2 d_R}{\lambda_{\min}(\gamma_R) \, \alpha \, q^2 \, \| Y \|_q^{q-1}} \, \epsilon,
\eee
and since $\epsilon$ may be arbitrarily small, we deduce that
\bee
\limsup_{t \rightarrow 0} | \Delta(t) | = \lim_{t \rightarrow 0} | \Delta(t) |  = 0.
\eee

%****************************************************************************
\section{Proof of Theorem \ref{thm:CB-HC-LS2}}

We prove parts (i) and (ii) of the theorem separately. 

\subsection{Proof of (i)}

We need to show that~\eqref{CB-logSob2} holds for any positive semidefinite $Y_{RS}$. 
A continuity argument (using the Fannes inequality~\cite{Fannes}) 
verifies that it suffices to prove~\eqref{CB-logSob2} for positive definite $Y_{RS}$. 
For this we apply Theorem~\ref{thm:der-NC-norm} with
\bee
X_{RS}( t) = ({\cal I}_R \ot \Phi_{t}) (Y_{RS}).
\eee
Since $Y_{RS}$ is positive definite, by Lemma~\ref{lem:positive-definite},
proved in Appendix~\ref{app:positive-definite}, we deduce that $X_{RS}(t)$ is also positive definite for all $t\geq 0$. 
We note that $X_{RS}(0) = Y_{RS}$ and
\bee
\frac{\dd }{\dd t} X_{RS}(t) = -  ({\cal I}_R \ot {\cal L}) (X_{RS}(t)),
\eee 
which gives
$X'_{RS}(0) = -  ({\cal I}_R \ot {\cal L})(Y_{RS}).$

Since by assumption ${\| \Phi_t \|}_{CB, q \rightarrow p(t)}  \le 1$ we have
\bee
\| X_{RS}( t) \|_{(q,p)} \le \| Y_{RS} \|_{q},
\eee
for all $t$ in a neighborhood of  $0$. 
Since equality holds at $t=0$, the derivative of $\| X_{RS}( t) \|_{(q,p)}$ 
at $t=0$ must be less than or equal to zero.
Then from Theorem~\ref{thm:der-NC-norm} we immediately conclude
\bee
\tau (Y^q \ln Y^q) - \tau_R \Big(\tau_S(Y^q) \ln \tau_S(Y^q)\Big) 
 - \alpha q^2 \tau (Y^{q-1} ({\cal I}_R \ot {\cal L}) (Y_{RS})) \le 0,
\eee
where as usual $\alpha = p'(0)^{-1}$.
%%%%%%%

\subsection{Proof of (ii)}

Our goal is to show that for any $Y_{RS}>0$ and $t\geq 0$ we have 
$\|\cI_R\ot \Phi_t(Y_{RS})\|_{(q, p(t))}\leq \| Y_{RS} \|_q$. 
Without loss of generality we assume that
\bee
\| Y_{RS} \|_q = 1,
\eee
so that our goal becomes $\|\cI_R\ot \Phi_t(Y_{RS})\|_{(q, p(t))}\leq 1$.
We assume that the CB log-Sobolev inequality holds for all $t \ge 0$,
with constant $\alpha = p'(t)^{-1}$.
We will argue by contradiction, so
let us suppose that 
\begin{align}\label{eq:t-0-c}
\|\cI_R\ot \Phi_{t_0}(Y_{RS})\|_{(q, p(t_0))}> 1,
\end{align}
for some $t_0>0$.
We will apply the results of Section~\ref{sec:q-p-norm} with
$$X_{RS}(t) = \cI_R\ot \Phi_t(Y_{RS}).$$
Note that by Lemma~\ref{lem:positive-definite}, $X_{RS}(t)$ is positive definite for all $t\geq 0$,
since by assumption $Y_{RS}$ is positive definite.

Define 
$$\widetilde \varphi(t) = \|X_{RS}(t)\|_{(q, p(t))} - \epsilon t.$$
Then by~\eqref{eq:t-0-c} for  sufficiently small $\epsilon>0$ we have
$$\widetilde\varphi(t_0)>1.$$
Let 
$$U=\{t \in[0, t_0]:\,   \widetilde\varphi(t)\leq 1 \}.$$
Since $\Phi_0=\cI_S$ and $p(0)=q$, we have $\widetilde \varphi(0)=1$ and thus $U$ is non-empty. 
Let $u=\sup U$. By Lemma~\ref{lem:continuous} the function $\widetilde\varphi(t)$ is continuous, 
so $u\in U$ and $\widetilde\varphi(u)\leq 1$. This means that $u< t_0$. Moreover, for any $t\in (u, t_0]$ we have
$$\widetilde\varphi(t)>1\geq \widetilde\varphi(u).$$

For $t>0$ let $\hat\sigma_R(t)$ be the unique minimizer characterized in Lemma~\ref{lem1}, and let $\hat\sigma_R(0)=\gamma_R$ where $\gamma_R$ is defined in~\eqref{eq:def-gamma}.
Define
$$\mu(t) = F(t, \hat\sigma_R(u)) -\epsilon t.$$
Then for any $t\geq u$ we have
$$\mu(t) \geq \inf_{\sigma_R} F(t, \sigma_R) -\epsilon t = 
\widetilde\varphi(t),$$
and we have $\widetilde\varphi(u)=\mu(u)$.

The derivative of $\mu(t)=F(t, \hat\sigma_R(u)) - \epsilon t$ at $t=u$ can be computed using the results of Lemma~\ref{lem:second-der},
and the characterization~\eqref{lem1:eqn2} of $\hat\sigma_R(u)$. The result is
\be
\frac{\partial }{\partial t} F(t, \hat\sigma_R(u))\Big|_{t=u} &=&
\frac{p'(u) F}{p^2 \, \tau (M^p)} \, \bigg[
 \tau(M^p \ln M^p) -  \tau_R \Big( \tau_S(M^p) \, \ln \tau_S(M^p)\Big) \nonumber \\
 && \hskip1.5in  -
 \frac{p^2}{p'(u)} \,\tau \left( M^{p-1} \, ({\cal I}_R \otimes {\mathcal L}) (M) \right) \bigg],
\ee
where $M = (\hat\sigma_R(u)^{-s(u)/2} \ot I_S) X_{RS}(u) (\hat\sigma_R(u)^{-s(u)/2} \ot I_S)$.
Then using the assumption that the semigroup satisfies the CB~log-Sobolev inequality at $p(u)$ with constant $\alpha(u)=1/p'(u)$, we find that 
$$\mu'(u) = \frac{\partial }{\partial t} F(t, \hat\sigma_R(u))\Big|_{t=u} - \epsilon \leq  -\epsilon.$$
Therefore there exists $\delta>0$ such that $u+\delta\leq t_0$ and $\mu(u+\delta)\leq  \mu(u)$. We then have
$$\widetilde\varphi(u+\delta)\leq \mu(u+\delta)\leq \mu(u)=\widetilde \varphi(u)\leq 1.$$
This contradicts the definition of $u$, therefore we conclude that the assumption \eqref{eq:t-0-c} is false,
and this establishes the Theorem.

%%%%%%%%%%%%%%%%%%%%%%%%%%%%%%%%%%

\section{Proof of Theorem \ref{thm:CB-LS-Gross}}

We suppose that the CB~log-Sobolev inequality holds at $q=2$ with constant $\alpha$, thus for any
positive semidefinite $Y=Y_{RS}$ we have
\be\label{GL-1}
\tau(Y^2 \ln Y^2) -  \tau_R \Big( \tau_S(Y^2) \, \ln \tau_S(Y^2)\Big) \le
4 \, \alpha \,\tau \left( Y \, ({\cal I}_R \otimes {\mathcal L}) (Y) \right).
\ee
We will prove that for any positive semidefinite $V$ and $q \ge 2$,
\be\label{GL-2}
4 \, \tau \left( V^{q/2} \, ({\cal I}_R \otimes {\mathcal L}) (V^{q/2}) \right)
\le \frac{q^2}{q-1} \, \tau \left( V^{q-1} \, ({\cal I}_R \otimes {\mathcal L}) (V) \right).
\ee
Letting $Y = V^{q/2}$ and combining the inequalities \eqref{GL-1} and \eqref{GL-2} we obtain the bound
\bee
\frac{\alpha}{q-1} \, q^2 \, \tau \left( V^{q-1} \, ({\cal I}_R \otimes {\mathcal L}) (V) \right)
\ge
 \tau(V^q \ln V^q) -  \tau_R \Big( \tau_S(V^q) \, \ln \tau_S(V^q)\Big)
\eee
which is precisely the CB-log Sobolev inequality with constant $\alpha (q-1)^{-1}$.

To prove~\eqref{GL-2} we will follow the method used in the recent paper \cite{CKMT},
which is itself based on the Stroock-Varopoulos inequality~\cite{S,V}.
The following inequality is proved in Appendix~\ref{app:GSV}: for all
positive definite $Z_{RS}$ and all $2 \le r \le q$,
\be\label{S-V-ineq}
r r' \, \tau \left( Z^{1/r} \, ({\cal I}_R \otimes {\mathcal L}) (Z^{1/r'}) \right) \le
q q' \, \tau \left( Z^{1/q} \, ({\cal I}_R \otimes {\mathcal L}) (Z^{1/q'}) \right), 
\ee
where $r', q'$ are the usual conjugate values defined by
\bee
\frac{1}{r'} = 1 - \frac{1}{r}, \qquad
\frac{1}{q'} = 1 - \frac{1}{q}.
\eee
The inequality \eqref{GL-2} follows by taking $r=2$, and $Z = V^q$ using the fact that $\mathcal L$ is self-adjoint.

\section{Proof of Theorem~\ref{thm:tensorization}}
The tensorization property of the CB log-Sobolev inequality can be proved using our main result Theorem~\ref{thm:CB-HC-LS2}, and the multiplicativity of the CB norm for completely positive maps~\eqref{eq:CB-norm-multiplicative}. Here we present a direct proof. 

Let $Y_{RS_1\dots S_k}$ be an arbitrary positive semidefinite matrix. By assumption for every $i$ we have 
$$ \tau(Y^q \ln Y^q) -  \tau \Big( \tau_{S_i}(Y^q) \, \ln \tau_{S_i}(Y^q)\Big) \le 
 {\alpha}{q^2}\,\tau \left( Y^{q-1} \, ({\cal I}_R \otimes \hat{\mathcal L}_i) (Y) \right).
$$
Then the claim follows if we show that 
\begin{align}\label{eq:sub-modular}
\sum_{i=1}^k \tau(Y^q \ln Y^q) -  \tau \Big( \tau_{S_i}(Y^q) \, \ln \tau_{S_i}(Y^q)\Big)\geq \tau(Y^q\ln Y^q) - \tau \Big( \tau_{S_1\dots S_k}(Y^q) \, \ln \tau_{S_1\dots S_k}(Y^q)\Big).
\end{align}

Recall that the conditional entropy of $\rho_{AB}$ with $\tr \rho_{AB}=1$ is defined by 
\begin{align*}
H(A|B) & = -\tr(\rho_{AB}\ln \rho_{AB}) + \tr_B\big(\tr_A (\rho_{AB}) \ln\tr_A(\rho_{AB} )\big)\\
&=-d_{AB}\tau(\rho_{AB}\ln \rho_{AB}) + d_{AB}\tau_B\big(\tau_A (\rho_{AB}) \ln\tau_A(\rho_{AB} )\big) -  \ln d_A,
\end{align*}
and satisfies the chain rule $H(AC|B)= H(A|B) + H(C|AB)$. Moreover, by the strong data processing inequality we have
$$H(A|B) \geq H(A|BC).$$

Observe that in~\eqref{eq:sub-modular} with no loss of generality we may assume that $Y^q$ is normalized as $\tr Y^q=1$. Then this inequality can be rewritten as 
\begin{align}\label{eq:modular-2}
\sum_{i=1}^k H(S_i| R S_{\sim i}) \leq H(S_1\dots S_k|R),
\end{align}
where we use $S_{\sim i} = S_1\dots S_{i-1} S_{i+1} \dots S_k$. Now using the chain rule we have 
\begin{align*}
H(S_1\dots S_k|R) = \sum_{i=1}^k  H(S_i| R S_1\dots S_{i-1}).
\end{align*}
On the other hand the strong data processing inequality gives
$$H(S_i| R S_{\sim i})\leq H(S_i| R S_1\dots S_{i-1}).$$
Using this inequality in the previous equation we arrive at~\eqref{eq:modular-2}.

%******************************************************************************************

\section*{Acknowledgements}
This  work  was  initiated  at  the  BIRS  workshop  15w5098,  ``Hypercontractivity  and  Log  Sobolev
Inequalities in Quantum Information Theory''.  We thank BIRS and the Banff Centre
for their hospitality. SB was supported in part by Institute of Network Coding of CUHK and by GRF grants 2150829 and 2150785.

%******************************************************************************************
%*********************************APPENDIX********************************************
%******************************************************************************************

\appendix
\vspace{.2in}
\begin{center}
{\Huge Appendix}
\end{center}
\section{Proof of Lemma~\ref{lem:second-der}} \label{app:proof-lem-second-der}
(a)
Since 
\bee
F(t, \sigma_R) = \left( \tau (M^p) \right)^{1/p},
\eee
it is sufficient to prove that $t\mapsto \tau (M^p)$ is twice continuously differentiable. We note that
\be\label{rep-Mp}
M^p = \frac{1}{2 \pi i} \int_{\Gamma} \frac{z^p}{z - M} \,  \dd z,
\ee
where $\Gamma$ is a closed contour which encloses the spectrum of $M$, which can be assumed to be in the open right half plane since $M$ is positive definite. Moreover, the function $z^p = e^{p \ln z}$ is defined with a cut along the negative real axis,
and is analytic for ${\rm Re}(z) > 0$. So we are reduced to proving that $\tau ((z - M)^{-1})$ is twice continuously differentiable for all $z$ outside the spectrum of $M$.

Explicit calculation yields
\begin{align*}
\frac{\dd }{\dd t} \tau ((z - M)^{-1}) &= \tau ((z - M)^{-1} M' (z - M)^{-1}), \\
\frac{\dd^2 }{\dd t^2} \tau ((z - M)^{-1}) &= 2 \tau \left((z - M)^{-1} M' (z - M)^{-1} M' (z - M)^{-1}\right) \\
& \quad + 
\tau ((z - M)^{-1} M'' (z - M)^{-1}).
\end{align*}
Thus we need to show that $M'$ and $M''$ are continuous. For this we need to show that 
$\sigma_R^{-s( t)/2}$ and $X_{RS}( t)$ are each twice continuously differentiable. $X_{RS}( t)$ is twice continuously differentiable by assumption. For $\sigma_R^{-s( t)/2}$ we again use the representation
\bee
\sigma_R^{-s( t)/2} = 
 \frac{1}{2 \pi i} \int_{\Gamma'} \frac{z^{-s( t)/2}}{z - \sigma_R} \,  \dd z,
\eee
where $\Gamma'$ is some  contour which encloses the spectrum of $\sigma_R$ and is in the open right half plane. The proof finishes observing that the function $z^{-s(t)/2} = e^{-s( t) \ln z/2}$ is analytic in $s$, and then twice continuously differentiable in $t$.

\medskip
\noindent{(b)}
Let $\xi(u) = \sigma_R + u A$ where $A=A^*$ is a self-adjoint matrix. Since $\sigma_R \in {\cD}^{+}_{R}$,
we have $\xi(u) \in {\cD}^{+}_{R}$ for $|u|$ sufficiently small. We define
$h(u) = F(t, \xi(u))$. Following the reasoning from the proof of part (a), it is sufficient to prove differentiability of
$\tau(z - M)^{-1}$, which boils down to differentiability of $\xi^{-s/2}$.
Using again the representation
\bee
\xi(u)^{-s( t)/2} = 
 \frac{1}{2 \pi i} \int_{\Gamma'} \frac{z^{-s( t)/2}}{z - \sigma_R - u A} \,  \dd z,
\eee
we see that $\xi(u)^{-s( t)/2}$ is analytic in $u$, which implies the desired result.

\medskip
\noindent{(c)}
Using the representation~\eqref{rep-Mp} we have 
\be
\frac{\dd }{\dd t} M^p = \frac{1}{2 \pi i} \int_{\Gamma} \left( p'(t)\frac{z^p\ln z}{z - M} + 
z^p \, \frac{1}{z - M} M' \, \frac{1}{z - M} \right) \, \dd z.
\ee
Let $\{(\lambda_i, v_i):\, 1\leq i\leq d_{RS}\}$ be the set of eigenvalues and eigenvectors of $M$. Then since $M$ is positive definite we have $M=\sum_i \lambda_i v_i v_i^*$ and 
\bee
\tau \Big(\frac{1}{z - M} M' \, \frac{1}{z - M}\Big) = d_{RS}^{-1} \, \sum_i (z - \lambda_i)^{-2} \,  v_i^{*}  M' v_i.
\eee
Now by the residue theorem
\bee
\frac{1}{2 \pi i} \int_{\Gamma} \frac{z^p}{(z - \lambda_i)^2} \dd z =&
p \, \lambda_i^{p-1},
\eee
and therefore
\begin{align}
\frac{1}{2 \pi i} \int_{\Gamma} z^p \, \tau \left(\frac{1}{z - M} M' \, \frac{1}{z - M} \right) \, \dd z & =
d_{RS}^{-1} \, \sum_i \,  p \, \lambda_i^{p-1}\,v_i^{*} M' v_i  \nonumber\\
& = p \, \tau (M^{p-1} \,  M' ).
\end{align}
Putting these together we arrive at
\bee
\frac{\dd }{\dd t} \tau( M^p) = p'(t)\tau (M^p \ln (M)) + p \, \tau (M^{p-1} M').
\eee
Therefore,
\begin{align*}
\frac{\partial}{\partial t} F(t, \sigma_R) &= F(t, \sigma_R)  \, \Big( - \frac{p'(t)}{p^2} \ln \tau (M^p) + \frac{1}{p \, \tau (M^p)} \frac{\dd }{\dd t} \tau( M^p) \Big) \\
&=
\frac{F(t, \sigma_R)}{p^2 \, \tau (M^p)} \, \Big( -p'(t) \tau (M^p) \, \ln \tau (M^p) + p'(t)p \,  \tau (M^p \ln (M)) +
p^2 \, \tau (M^{p-1} M') \Big).
\end{align*}
Finally we compute the derivative of $M$. First we note that
\bee
\frac{\dd }{\dd t} \sigma_R^{-s( t)/2} = - \frac{s'( t)}{2} \sigma_R^{-s( t)/2}\ln \sigma_R.
\eee
To justify this equation we may assume without loss of generality that $\sigma_R$ is diagonal. Since $s'( t) = p'(t)p^{-2}$ it follows that
\begin{align*}
\frac{\dd }{\dd t} M =&
- \frac{p'(t)p^{-2}}{2} (\ln \sigma_R \ot I_S) \, M  - \frac{p'(t)p^{-2}}{2} M \, (\ln \sigma_R \ot I_S)
\\
&+ (\sigma_R^{-s( t)/2} \ot I_S) X_{RS}'( t) (\sigma_R^{-s( t)/2} \ot I_S).
\end{align*}
Thus we find
$$
\tau (M^{p-1} M') = - p'(t)p^{-2} \, \tau \big( M^p \, (\ln \sigma_R \ot I_S)\big) +
\tau \big(M^{p-1} \, (\sigma_R^{-s( t)/2} \ot I_S) X_{RS}'( t) (\sigma_R^{-s( t)/2} \ot I_S)\big).
$$
Combining these and using $ \tau \big( M^p (\ln \sigma_R \ot I_S)\big) = \tau_R \big(\tau_S(M^p) \ln \sigma_R\big)$ we get
\begin{align*}
\frac{\partial}{\partial t} F(t,\sigma_R) &= \frac{p'(t)\, F(t,\sigma_R)}{p^2 \, \tau (M^p)} \, \bigg[
 - \tau (M^p) \,\ln \tau (M^p)+  \tau (M^p \ln M^p)   -  \tau_R \left( \tau_S(M^p) \ln \sigma_R \right) \nonumber\\
& \hskip 1.3in+  \frac{p^2}{p'(t)} \, \tau \left( M^{p-1} (\sigma_R^{-s( t)/2} \ot I_S) X_{RS}'( t) (\sigma_R^{-s( t)/2} \ot I_S)  \right) \bigg].
\end{align*}
Also $M(0) = X_{RS}(0) = Y$, and $p(0)=q$. Using these in the above equation gives~\eqref{eq:deriv-t=0}. 

%********************************************************************************************

\section{Lipschitz constant of the relative entropy function }\label{app:rel-ent-Lip}

Here we provide some estimates for the Lipschitz constant of the relative entropy function. As before,
we will denote by $\lambda_{\min}(\sigma)$ the smallest eigenvalue of $\sigma \in {\cD}_R^{+}$.

\begin{lemma}\label{lem:rel-ent-Lip}
Let $\gamma, \sigma, \xi \in {\cD}_R^{+}$ be such that $\|\gamma-\sigma\|_1<\kappa$ and $ \|\gamma-\xi\|_1<\kappa$ where
\bee
\kappa = \frac{1}{2 d_R} \lambda_{\min}(\gamma).
\eee
Then we have
\be
\Big| S(d_R^{-1}\gamma \| d_R^{-1}\sigma) - S(d_R^{-1}\gamma \| d_R^{-1}\xi) \Big|
\le 4\kappa \, \| \sigma - \xi \|_1.
\ee
\end{lemma}

\begin{proof}
Suppose that $\lambda_{\min}(\sigma) \geq \lambda_{\min}(\gamma)$, and that $v,w$ are respectively the
normalized eigenvectors of $\sigma$ and $\gamma$ for these eigenvalues. Then
\bee
| \lambda_{\min}(\sigma) - \lambda_{\min}(\gamma) | &=& \lambda_{\min}(\sigma) - \lambda_{\min}(\gamma)  \\
&=& v^* \sigma v - w^* \gamma w  \\
&\le& w^* ( \sigma  -  \gamma ) w  \\
& \le & \| \sigma - \gamma \|_{\infty}  \\
& \le & d_R \, \| \sigma - \gamma \|_1  \\
& < &  d_R \, \kappa.
\eee
The same bound holds if $\lambda_{\min}(\sigma) \leq \lambda_{\min}(\gamma)$. We similarly have $| \lambda_{\min}(\xi) - \lambda_{\min}(\gamma) |<d_R\kappa$. As a result we have
\be\label{bound-min-diff}
\lambda_{\min}(\sigma), \lambda_{\min}(\xi)\geq \frac{1}{2}\lambda_{\min}(\gamma).
\ee

By definition
\bee
S(d_R^{-1}\gamma \| d_R^{-1}\sigma) - S(d_R^{-1}\gamma \| d_R^{-1}\xi)= \tau \big( \gamma (\ln \xi - \ln \sigma) \big),
\eee
and therefore 
\be\label{S-bound2}
\Big| S(d_R^{-1}\gamma \| d_R^{-1}\sigma) - S(d_R^{-1}\gamma \| d_R^{-1}\xi) \Big|
\le  \| \ln \sigma - \ln \xi \|_{\infty}.
\ee
Furthermore,
\bee
\ln \sigma - \ln \xi &=& \int_{0}^{\infty} \left( \frac{1}{t + \xi} - \frac{1}{t + \sigma} \right) \, \dd t \\
&=& \int_{0}^{\infty} \left( \frac{1}{t + \sigma} (\sigma - \xi) \frac{1}{t + \xi} \right) \, \dd t.
\eee
Hence
\bee
\| \ln \sigma - \ln \theta \|_{\infty} &\le &
\int_{0}^{\infty} \left( \frac{1}{t + \lambda_{\min}(\sigma)} \| \sigma - \xi \|_{\infty} \frac{1}{t + \lambda_{\min}(\xi)} \right) \, \dd t \\
&\le& \| \sigma - \xi \|_{\infty} \, \int_{0}^{\infty} \Big(t + \frac 1 2\lambda_{\min}(\gamma)\Big)^{-2} \, \dd t \\
& =& \| \sigma - \xi \|_{\infty} \, \frac{2}{\lambda_{\min}(\gamma)} \\
& \le & \| \sigma - \xi \|_{1} \, \frac{2 d_R}{\lambda_{\min}(\gamma)}. 
\eee
where we used~\eqref{bound-min-diff}. Substituting this into~\eqref{S-bound2} we get the desired bound.
\end{proof}

%***************************************************************************************
\section{Strict positivity of $\cI_R\ot \Phi_t$}\label{app:positive-definite}

\begin{lemma}\label{lem:positive-definite}
If $Y_{RS}$ is positive definite, then $\cI_R\otimes \Phi_t(Y_{RS})$ is positive definite for all $t\geq 0$.
\end{lemma}

\begin{proof} 
By assumption $\cI_R\otimes \Phi_t(Y_{RS})$ is positive semidefinite. Then if it is not positive definite, it must be singular. That is, there is $0\neq v\in \cH_{RS}$ such that $v^*\cI_R\otimes \Phi_t(Y_{RS})v=0$. Let $Z_{RS}\geq 0$ be an arbitrary positive semidefinite matrix. Since $Y_{RS}$ is positive definite, $Y_{RS}- \epsilon Z_{RS}\geq 0$ for sufficiently small $\epsilon>0$. Hence, 
$$\cI_R\ot \Phi_t(Y_{RS}) \geq \epsilon\, \cI_R\ot \Phi_t(Z_{RS}) \geq 0,$$
and then 
$$v^*\cI_R\ot \Phi_t(Y_{RS})v \geq \epsilon\, v^*\cI_R\ot \Phi_t(Z_{RS}) v\geq 0,$$
which gives $v^*\cI_R\ot \Phi_t(Z_{RS}) v=0$ for all $Z_{RS}\geq 0$. Let 
$$Z_{RS}=\cI_R\ot e^{t\mathcal {L}} (I_{RS}).$$
Then using~\eqref{eq:Phi-exp} we find that $v^*v=0$ which is a contradiction since $v\neq 0$. Therefore, $\cI_R\otimes \Phi_t(Y_{RS})$ is positive definite for all $t\geq 0$.

\end{proof}

%*******************************************************************************************
\section{The quantum Gross Lemma}\label{app:GSV}

Let $Z_{RS}$ be positive definite with spectral decomposition
\bee
Z_{RS} = \sum_i \lambda_i \, w_i w_i^*.
\eee
Then for any $a,b \in \mathbb{R}$ we have
\bee
\tau \left( Z^a \, ({\cal I}_R \otimes {\mathcal L}) (Z^b) \right) =
\sum_{i,j}  \lambda_i^a \, \lambda_j^b \, L_{ij},
\eee
where
\bee
L_{ij} = \tau \left( w_i w_i^* ({\cal I}_R \otimes {\mathcal L}) (w_j w_j^*) \right).
\eee
Since the semigroup is  reversible and $\mathcal L^*=\mathcal L$, we have $L_{ij} = L_{ji}$ for all $i,j$. Moreover, since the semigroup is unital, we have $\mathcal L(I_S) =0$ and 
\bee
\sum_j L_{ij} = 0 \quad \text{for all $i$}.
\eee
Using these properties we can write
\be\label{GL-ineq-pf1}
\tau \left( Z^a \, ({\cal I}_R \otimes {\mathcal L}) (Z^b) \right) = 
- \frac{1}{2} \, \sum_{i,j}  (\lambda_i^a - \lambda_j^a) \, (\lambda_i^b -  \lambda_j^b) \, L_{ij}.
\ee
On the other hand $\Phi_t = e^{- t {\mathcal L}}$ is completely positive for $t \ge 0$, so in particular
\bee
\tau \left( w_i w_i^* ({\cal I}_R \otimes \Phi_t) (w_j w_j^*) \right)
= - t \, L_{ij} + O(t^2) \ge 0,
\eee
for all $i \neq j$. Thus $L_{ij} \le 0$ for all $i \neq j$.

We apply the representation \eqref{GL-ineq-pf1} to \eqref{S-V-ineq} on the left side with $a=1/r$ and $b=1/r'$, and on the right side with
$a=1/q$ and $b=1/q'$. It is sufficient to prove the inequality for each index pair $i \neq j$:
\be\label{GL-ineq-pf2}
r r' \, (\lambda_i^{1/r} - \lambda_j^{1/r}) \, (\lambda_i^{1/r'} -  \lambda_j^{1/r'}) \le
q q' \, (\lambda_i^{1/q} - \lambda_j^{1/q}) \, (\lambda_i^{1/q'} -  \lambda_j^{1/q'}).
\ee
We assume without loss of generality that $\lambda_i> \lambda_j$ and let $c = \lambda_i/\lambda_j>1$. Define
\bee
f(u) = \frac{c^u - 1}{u}.
\eee
Then the left side of~\eqref{GL-ineq-pf2} is $\lambda_jf(1/r) f(1- 1/r)$. Since $1/q < 1/r$ and $\lambda_j>0 $,
the bound will follow if we can show that for all
$0 < u < v \le 1/2$ we have
\bee
f(u) f(1-u) \ge f(v) f(1-v).
\eee
Equivalently, we can show that
\bee
\log f(u) + \log f(1-u) \ge \log f(v) + \log f(1-v).
\eee
The function $g(u)=\log f(u) + \log f(1-u)$ is symmetric around $u = 1/2$, so the above inequality follows if we show that it is convex. For this it is sufficient to show that $\log f(u)$ is convex, and
this follows from a straightforward calculation of its second derivative.

%*******************************************************************************************
{~~}

\end{document}